\documentclass[conference,letterpaper]{IEEEtran}

\usepackage[utf8]{inputenc}
\usepackage[T1]{fontenc}
\usepackage{url}
\usepackage{ifthen}
\usepackage{cite}
\usepackage[cmex10]{amsmath} 
\usepackage{algorithm,graphicx}
\usepackage{algpseudocode}
\usepackage{bbm}
\usepackage{mathrsfs}
\usepackage{amsfonts}
\usepackage{subfigure}

\newtheorem{definition}{Definition}
\newtheorem{theorem}{Theorem}
\newtheorem{proof}{Proof}

\newcommand\bL{{\bf L}}

\newcommand\be{{\bf e}}

\newcommand\bv{{\bf v}}
\newcommand\bz{{\bf z}}
\newcommand\by{{\bf y}}

\begin{document}
\title{Pruned Collapsed Projection-Aggregation Decoding of Reed-Muller Codes}

 \author{%
   \IEEEauthorblockN{Qin Huang*, \emph{Senior Member}, Bin Zhang}
   \IEEEauthorblockA{Department of Electronic and Information Engineering, Beihang University, Beijing, China, 100191\\
                    Email: qhuang.smash@gmail.com; benzhang@buaa.edu.cn}
 }

%
%


\maketitle

\begin{abstract}
The paper proposes to decode Reed-Muller (RM) codes by projecting onto only a few subspaces such that the number of projections is significantly reduced. It reveals that the probability that error pairs are canceled simultaneously in two different projections is determined by their intersection size. Then, correlation coefficient which indicates the intersection size of two subspaces in a collection is introduced for collecting subspaces. Simulation results show that our proposed approach with a small number of projections onto collected subspaces performs close to the original approach.
\end{abstract}


\renewcommand{\thefootnote}{}
\footnotetext{This work was supported by the National Natural Science Foundation of China under Grant 62071026 and 61941106.  \textit{(Q. Huang and B. Zhang contributed equally to this work.)}}

\section{Introduction}

Reed-Muller (RM) codes which can be designed from the set of all the $m$-variate Boolean polynomials of degree $r$ or less were first discovered by Muller \cite{muller1954application} and the first decoder was devised by Reed \cite{reed1953class}.
RM codes outperform polar codes under maximum a posteriori probability (MAP) decoding for general symmetric channels\cite{arikan2009channel},\cite{hussami2009performance}. However, the MAP decoding is impractical due to its overwhelming complexity as code length increases. Thus, it is  an important question to decode RM code with practical complexity.

In 2004, Dumer \cite{dumer2004recursive} decomposed a RM($r$,$m$) code into trivial codes, such as repetition codes and codes with rate 1. The resulted recursive list decoding \cite{dumer2004recursive} and \cite{dumer2006soft} performs very well on short RM codes. Recent years, many approaches exploited the large automorphism groups \cite{macwilliams1977theory} of RM codes \cite{ye2020recursive}, \cite{lian2020decoding}, \cite{santi2018decoding} and \cite{hashemi2018decoding}. Among these, \emph{recursive projection-aggregation} (RPA) decoding \cite{ye2020recursive} achieves excellent performance on low-rate RM codes. RPA projects the received vector onto all the one-dimensional subspaces of the vector space $\mathbb{F}_2^m$ over binary field in a recursive way, and then obtains a large number $2^{m-1}$ of trivial codes. Decode these trivial codes, and then aggregate these estimates to make a decision on the received vector. Unlike the recursive manner of RPA, \emph{collapsed projection-aggregation} (CPA) decoding directly projects the received vector onto all the $(r-1)$ dimensional subspaces \cite{lian2020decoding}. After decoding $N_{r,m}=\binom{m}{r-1}_2=\prod_{i=1}^{r-1} \frac{2^{m-i+1}-1}{2^{r-i}-1}$ different RM$(1,m-r+1)$ codes of these subspaces, $N_{r,m}$ estimates are aggregated to make a decision.

In order to reduce the number of projections, this paper proposes to project only a few subspaces from the $N_{r,m}$ subspaces, and come up with a \emph{pruned collapsed projection-aggregation} (PCPA) algorithm. In this paper, we find that a pair of errors in the received vector will be canceled each other out in the projection of a subspace, if their indices belong to the same coset of the subspace. Then we prove that the probability that error pairs are canceled simultaneously in two different projections, is determined by the intersection size of the two subspaces. Thus, we define correlation coefficient for collecting subspaces, which indicates the intersection size of two subspaces in a collection. Our analysis verifies that for two different subspaces, the larger correlation coefficient is, the more possible their estimates coincide. As a result, we propose to decode RM codes by projecting the collected subspaces with small correlation coefficients. Simulation results show that under the same size, subspace collections with smaller correlation coefficients outperform those with larger correlation coefficients. Moreover, our proposed PCPA with $64$ projections performs close to CPA with $2667$ projections and RPA with $8001$ projections when decoding RM$(3,7)$ code. As a result, it provides a good trade-off between performance and complexity.

The rest of this paper is organized as follows. Section II presents the necessary background. Our proposed PCPA decoding as well as the analysis of correlation coefficients is described in Section III. Section IV gives our simulation results and makes a comparison. Section V concludes this paper.

\section{Preliminaries}

For any integers $r$ and $m$ with $0 \le r \le m$, define $\mathbb{E}:=\mathbb{F}_2^m$.
Let $\bz = (z_1,z_2,...,z_m)$ denote a vector of $\mathbb{E}$ and $f=f(z_1,z_2,...,z_m)$  denote a $m$-variate polynomial Boolean function of degree at most $r$.
The $r$-th order binary Reed-Muller code RM$(r,m)$ of length $n=2^m$ is the set of all vectors $\bv$ which is obtained
from $f$ by evaluating $f$ at all elements of $\mathbb{E}$. The components of $\bv = (v_0, v_1, ..., v_{2^m-1})$ are indexed by $\bz \in \mathbb{E}$.

Let $\mathbb{B}$ be an $s$-dimensional subspace of $\mathbb{E}$, $0 \le s\le r$, and the coset of $\mathbb{B}$ in the quotient space $\mathbb{E}/\mathbb{B}$ is denoted by $T$. For a binary vector $\bv=(\bv(\bz),\bz \in \mathbb{E})$, its projection on the cosets of $\mathbb{B}$ is
\begin{equation}
    \bv_{/ \mathbb{B}} = \text{Proj} (\bv, \mathbb{B}):=(v_{/\mathbb{B}}(T),T \in \mathbb{E}/\mathbb{B}),
\end{equation}
where
\begin{equation}\label{eq:vT}
    v_{/\mathbb{B}}(T) := \underset{\bz \in T}{\oplus}v(\bz).
\end{equation}
It has been proved that if $\bv \in$ RM$(r,m)$, $\bv_{/ \mathbb{B}} \in$ RM$(r-s, m-s)$.

For a binary-input memoryless channel $W:\{0,1\}\to \mathcal{W} $, the
log-likelihood (LLR) ratio of the channel output $\by$ is
\begin{equation}\label{eq:llr}
    \bL(\by) = \ln( \frac{W(\by|0)}{W(\by|1)} ).
\end{equation}
The projection of a channel output vector $\bL(\by)$ of length $2^m$, where $L(y_i)$, $i=0,1,\ldots,2^m-1$, is given by (\ref{eq:llr}),  on the cosets of $\mathbb{B}$ is defined as
\begin{equation}
    \bL_{/ \mathbb{B}} :=\text{Proj} (\bL, \mathbb{B}):=(\bL_{/\mathbb{B}}(T),T \in \mathbb{E}/\mathbb{B}),
\end{equation}
where
\begin{equation}\label{eq:LT}
    \bL_{/\mathbb{B}}(T)=2 \tanh ^{-1}(\underset{\bz \in T}{\prod}
    \tanh(\frac{L(\bz)}{2})).
\end{equation}

CPA works in an iterative manner. Take the received $\bL$ as $\hat{\bL}^{(0)}$.
In the $i$-th iteration, project $\hat{\bL}^{(i-1)}$ onto all the $(r-1)$-dimensional subspaces of $\mathbb{E}$ and get $N_{r,m}=\binom{m}{r-1}_2=\prod_{i=1}^{r-1} \frac{2^{m-i+1}-1}{2^{r-i}-1}$ different RM$(1,m-r+1)$ LLR vectors. Decode these vectors using \emph{fast Hadamard transform} (FHT) and get $N_{r,m}$ estimates from the $N_{r,m}$ trivial codes. Using majority votes as shown in Algorithm 2, we aggregate all the estimates together with $\hat{\bL}^{(i-1)}$ to obtain a new estimate $\hat{\bL}^{(i)}$. Iterate above steps until reach the maximum number of iterations.

\section{Pruned Collapsed Projection-Aggregation Decoding}
\label{sec:ourworks}

When decoding a received vector of RM$(r,m)$, CPA projects the vector onto all the $(r-1)$-dimensional subspaces of $\mathbb{E}$. The large number of projections promises good performance, but brings huge computational complexity. In order to reduce the computational complexity, we come up with the proposed PCPA algorithm.

\subsection{Algorithm Description}
\label{subsec:PCPA}

Let $S$ denote the collection of subspaces from the $N_{r,m}$ subspaces and $\mathbb{B}_i$ is the $i$-th subspace in $S$ where $i=1,2,..,|S|$. Our proposed PCPA only projects onto the $|S|$ $(r-1)$-dimensional subspaces of $\mathbb{E}$. We will explain the subspace collection in the next subsection. The rest of the decoding steps of PCPA mimic those of CPA. The pseudo code is as follows.

\begin{algorithm}
\caption{CPA and PCPA}

\textbf{Input:} The LLR vector $\bL$; the collection of $(r-1)$-dimensional subspaces; the maximal number of iterations $T_{max}$; a scaling factor $\omega_{r,m}$ used to optimize the algorithm.

\textbf{Output:} The decoded codeword $\hat{\bv}$

\begin{algorithmic}
\State $\hat{\bL}^{(0)} = \bL$
\For {$i=1,2,\dots,T_{\max}$}

\State $L_{/\mathbb{B}_j} \gets \texttt{proj} (\bL,\mathbb{B}_j)$ for all $\mathbb{B}_j \in S$
\State $\hat{\bv}_{/\mathbb{B}_j} \gets \texttt{FHT} (\bL_{/\mathbb{B}_j})$ for all $\mathbb{B}_j \in S$
\State $\hat{\bL}^i \gets \texttt{Aggregation}
(\hat{\bL}^{(i-1)},\hat{\bv}_{/\mathbb{B}_1},\hat{\bv}_{/\mathbb{B}_2},...,\hat{\bv}_{/\mathbb{B}_{|S|}},\omega_{r,m})
$
\EndFor

\State $\hat{v}(\bz) \gets
\mathbbm{1}[ \hat{L}^{(N_{max})}(\bz)<0]$ for each $\bz \in \mathbb{E}$

\end{algorithmic}
\end{algorithm}

\begin{algorithm}
\caption{Aggregation}
\textbf{Input:} $\bL,\hat{\bv}_{/\mathbb{B}_1},\hat{\bv}_{/\mathbb{B}_2},...,\hat{\bv}_{/\mathbb{B}_{|S|}}
,\omega_{r,m}$

\textbf{Output:} $\hat{\bL}$

\begin{algorithmic}
\State $\hat{\bL}(z) \gets \omega_{r,m} \sum_{i=1}^{|S|}(1-2\hat{v}_{/\mathbb{B}_i}(T))L_{/\mathbb{B}}(T-\{\bz\})$
\State for each $\bz \in \mathbb{E}$ and $T$ is the coset of subspace $\mathbb{B}_i$ which contains $\bz$

\end{algorithmic}
\end{algorithm}

Note that if we use all the $(r-1)$-dimensional subspaces in $\mathbb{E}$, PCPA becomes CPA.

\subsection{Collecting Subspaces with Correlation Coefficients}
\label{subsec:coef}
Before giving a collecting guideline, we would like to discuss the correlation between the projections obtained from two different subspaces.

Consider two  subspaces $\mathbb{B}_1$ and $\mathbb{B}_2$, as well as their intersection $D=\mathbb{B}_1 \cap \mathbb{B}_2$. Every coset $T_1 = \bz'+\mathbb{B}_1$ of $\mathbb{B}_1$, where $\bz'\in\mathbb{E}$, can be rewritten as $T_1 = \bz'+D \cup (\mathbb{B}_1/D)$. So does the coset $T_2$ of $\mathbb{B}_2$: $T_2 = \bz'+D \cup (\mathbb{B}_2/D)$.
Taking these into (\ref{eq:vT}),
the projection $v(T_1)$ and $v(T_2)$ can be written as:

\begin{equation}\label{eq:d_b}
\begin{split}
	v_{/\mathbb{B}_1}(T_1) := [\underset{\bz \in \bz'+D}{\oplus}v(\bz)]\oplus
[\underset{\bz \in \bz'+\mathbb{B}_1/D}{\oplus}v(\bz)] \\
	v_{/\mathbb{B}_2}(T_2) := [\underset{\bz \in z'+D}{\oplus}v(\bz)]\oplus
[\underset{\bz \in \bz'+\mathbb{B}_2/D}{\oplus}v(\bz)]
\end{split}
\end{equation}

From the above formulas, we can see if paired errors occur at the entries which are indexed by $\bz$ where $\bz\in \bz'+D$,
they cancel each other out when projecting onto $T_1$ and $T_2$.
An example is shown in Fig. \ref{fig:eg1}.
Suppose that the first two entries of vector $\bv$  are corrupted.
After projecting $\bv$ onto two different subspaces
$\mathbb{B}_1$ and $\mathbb{B}_2$, whose intersection $D=\{(0\ 0\ 0\ 0),
(0\ 0\ 0\ 1),(0\ 0\ 1\ 0),(0\ 0\ 1\ 1)\}$,
the errors are canceled in the
projection $\bv_{/\mathbb{B}_1}$ and $\bv_{/\mathbb{B}_2}$, simultaneously.
In other words, there are no chances to detect the errors from the two projections.

However, if the two subspaces do not share the vectors
$\{(0\ 0\ 0\ 1),(0\ 0\ 1\ 0),(0\ 0\ 1\ 1)\}$, the errors will be detected.
Another example is shown in Fig. \ref{fig:eg2}, we respectively project $\bv$ onto
$\mathbb{B}_1$ and $\mathbb{B}_3$ whose intersection $D=\{(0\ 0\ 0\ 0)\}$. Although the errors are canceled in $\bv_{/\mathbb{B}_1}$, they are remained in $\bv_{/\mathbb{B}_3}$. If we could decode $\bv_{/\mathbb{B}_3}$ correctly, we will detect the errors. The following theorem demonstrates that the smaller intersection between two subspaces is, the more possible we can detect a pair of errors.

\begin{figure}[htbp]
   \centering
   \includegraphics[width=0.50\textwidth]{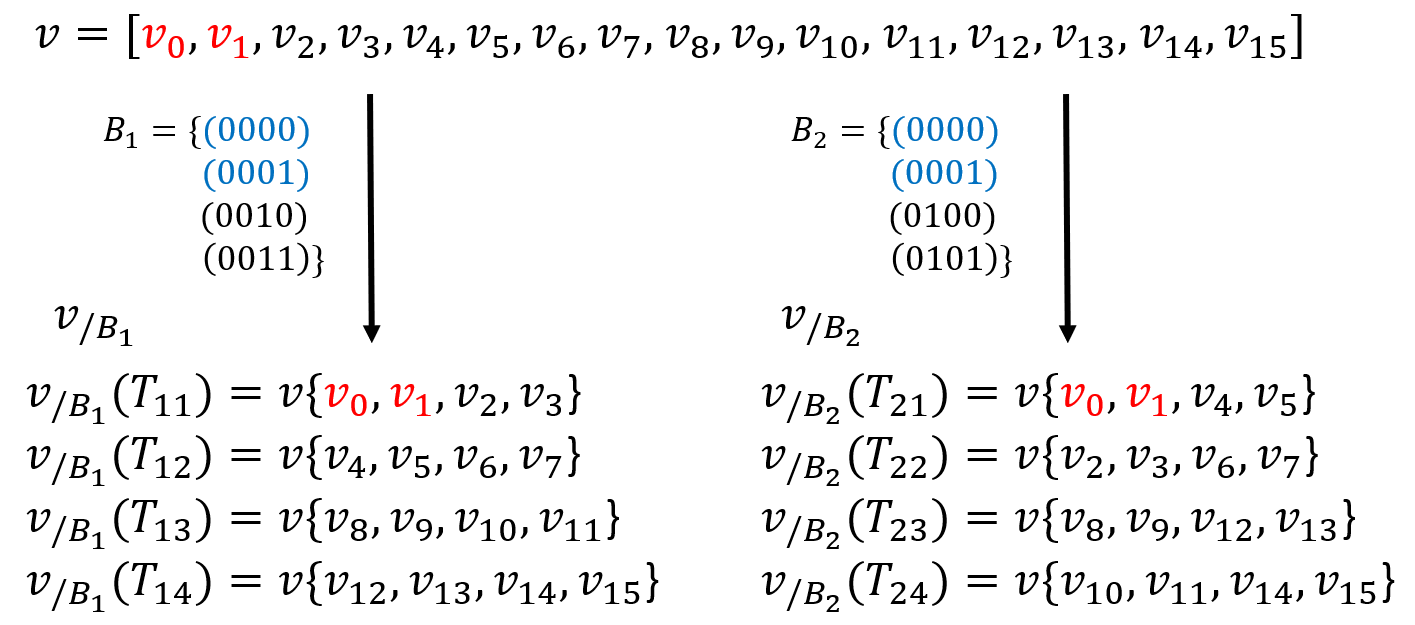}
   \caption{The projections $\bv_{/\mathbb{B}_1}$ and $\bv_{/\mathbb{B}_2}$ onto the cosets of $\mathbb{B}_1$ and $\mathbb{B}_2$. The corruptions are marked as red and the intersection of subspaces is marked as blue.}
   \label{fig:eg1}
\end{figure}

\begin{figure}[htbp]
   \centering
   \includegraphics[width=0.50\textwidth]{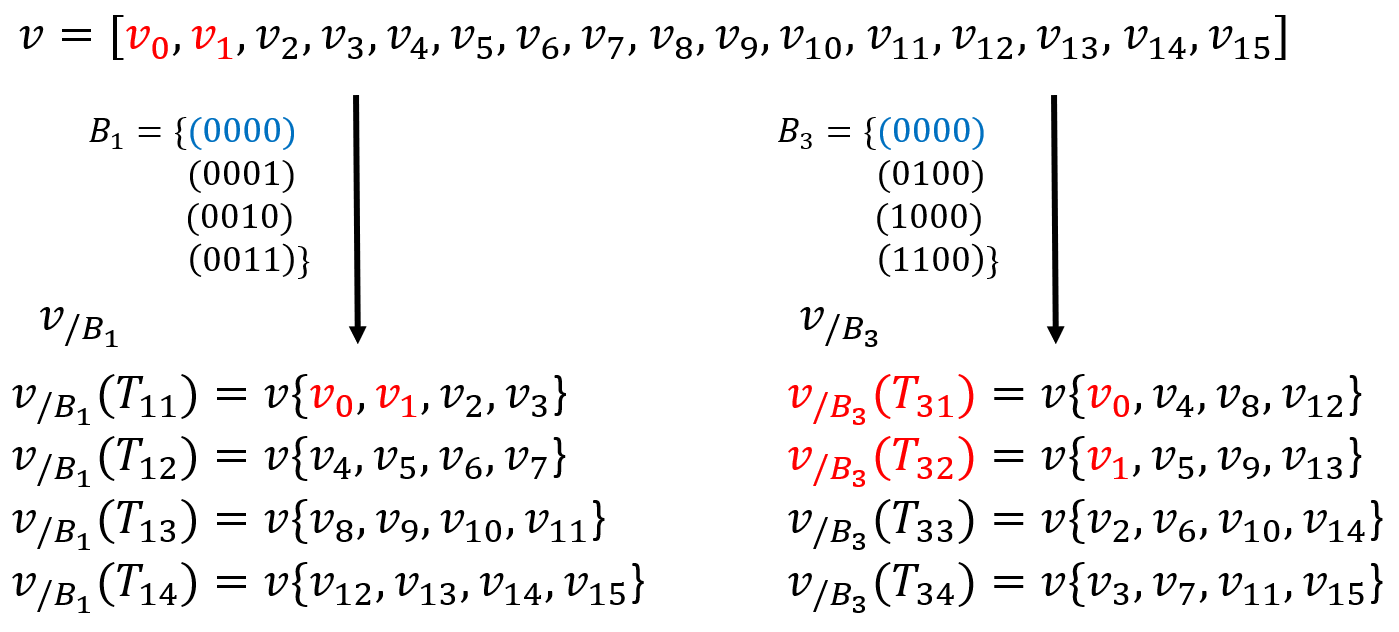}
   \caption{The projections $\bv_{/\mathbb{B}_1}$ and $\bv_{/\mathbb{B}_2}$ onto the cosets of $\mathbb{B}_1$ and $\mathbb{B}_3$. The corruptions are marked as red and the intersection of subspaces is marked as blue.}
   \label{fig:eg2}
\end{figure}

\begin{theorem}\label{theorem1}
Consider an error pattern $\be$ of length of $2^m$, where errors occur independently with probability $\epsilon$.  Project $\be$ onto $s$-dimensional subspaces $\mathbb{B}_1$ and $\mathbb{B}_2$, and obtain
$\be_{/\mathbb{B}_1}$ and $\be_{/\mathbb{B}_2}$. For any $\bz\in\mathbb{E}$, the probability that $\be_{/\mathbb{B}_1}(\bz+\mathbb{B}_1)=\be_{/\mathbb{B}_2}(\bz+\mathbb{B}_2)$ is

\begin{equation}
P=\frac{1}{2}[1+(1-2\epsilon)^{(2^{s+1}-2|\mathbb{B}_1 \cap \mathbb{B}_2|)}]
\end{equation}
\end{theorem}

\begin{proof}
The intersection of subspaces $\mathbb{B}_1$ and $\mathbb{B}_2$ is denoted by $D=\mathbb{B}_1 \cap \mathbb{B}_2$.
Note that every coset $T_1=\bz'+\mathbb{B}_1$ of $\mathbb{B}_1$ can be rewritten as $T_1=(\bz'+D) \cup (\bz'+ \mathbb{B}_1/D)$. So does the coset $T_2$ of $\mathbb{B}_2$: $T_2= (\bz'+D) \cup (\bz'+ \mathbb{B}_2/D)$.

Let $N_1$ denote the number of errors in $(e(\bz), \bz \in \bz'+\mathbb{B}_1/D)$
and $N_2$ denote the number of errors in $(e(\bz), \bz \in \bz'+\mathbb{B}_2/D)$.
Take $N_q=|\mathbb{B}_1/D|=|\mathbb{B}_2/D|
=2^s-|\mathbb{B}_1 \cap \mathbb{B}_2|$.
According to Equation
(\ref{eq:d_b}), $e_{/\mathbb{B}_1}(T_1)=e_{/\mathbb{B}_2}(T_2)$ if and only if $N_1$ and $N_2$ are both odd or even.

Because the errors occur independently, the probability $P_{12}$ that $N_1$ is even equals to the probability $P_{22}$ that $N_2$ is even:
\begin{equation}
\begin{split}
P_{12} = P_{22} &= \frac{1}{2}\sum_{i=0}^{N_q}(1+(-1)^i)\binom{N_q}{i}
\epsilon ^i (1-\epsilon)^{N_q-i}\\
&=\frac{1}{2}[1+(1-2\epsilon)^{N_q}]
\end{split}
\end{equation}
And the probability $P_{11}$ that $N_1$ is odd equals to the probability $P_{21}$ that $N_2$ is odd:
\begin{equation}
\begin{split}
P_{11} = P_{21} &= \frac{1}{2}\sum_{i=0}^{N_q}(1-(-1)^i)\binom{N_q}{i}
\epsilon ^i (1-\epsilon)^{N_q-i}\\
&=\frac{1}{2}[1-(1-2\epsilon)^{N_q}]
\end{split}
\end{equation}

The probability $P$ that $N_1$ and $N_2$ are both odd or even is:
\begin{equation}
\begin{split}
P&=P_{11}P_{21}+P_{12}P_{22}\\
&=\{\frac{1}{2}[1+(1-2\epsilon)^{N_q}]\}^2+
\{\frac{1}{2}[1-(1-2\epsilon)^{N_q}]\}^2\\
&=\frac{1}{2}[1+(1-2\epsilon)^{2N_q}].
\end{split}
\end{equation}
\end{proof}
This completes the proof.

As a result, we define the following correlation coefficients to help us to collect subspaces for projections.

\begin{definition}
Let $\mathbb{B}_i$ and $\mathbb{B}_j$, $i,j=1,2,\ldots,|S|$, be two $s$-dimensional subspaces in $\mathbb{E}$.
The correlation coefficient between $\mathbb{B}_i$ and $\mathbb{B}_j$ is defined as

\begin{equation}
    r_{ij} := \frac{\dim(\mathbb{B}_i \cap \mathbb{B}_j)}{s}
\end{equation}

\end{definition}

\begin{definition}
Let $S$ denote a collection of $s$-dimensional subspaces in $\mathbb{E}$.
$\mathbb{B}_i$ is the $i$-th subspace contained in $S$
for $i=1,2,...,|S|$.
The correlation coefficient of $S$ is defined as

\begin{equation}
    r_S := \sum_{i=1}^{|S|}\sum_{j=1}^{|S|}r_{ij}
\end{equation}

\end{definition}

To show the impact of correlation coefficients on projection results, we take the decoding of RM$(3,5)$ code as an example. There are $N_{3,5}=155$ two-dimensional subspaces for projections at each iteration while CPA decoding. We illustrate the correlation coefficients $r_{ij}$  in Fig. \ref{fig:first_case}. Both the $x$-axis and the $y$-axis represent the indices of the two-dimensional subspaces and the value of grid $(i,j)$ is $r_{ij}$, $i,j=1,2,...,155$. Since there are $155$ two-dimensional subspaces for projections, we have $155$ estimates at each iteration. We take $10000$ Monte Carlo trials at $2.0$dB over the binary input additive white Gaussian noise channel(BAWGNC) and count the frequency $p_{ij}$ that the $i$-th estimate and the $j$-th estimate are correct or false simultaneously at the first iteration for $i,j=1,2,...,155$. The frequency matrix $P=\{p_{ij}\}_{155\times 155}$ is shown  in Fig. \ref{fig:second_case}.
Both the $x$-axis and the $y$-axis represent the indices of the two-dimensional subspaces and the value of grid $(i,j)$ is $p_{ij}$. From the above two figures, for the $i$-th and the $j$-th subspaces, $i=1,2,\ldots,155$ and $j=1,2,\ldots,155$, we find that the larger correlation coefficient $r_{ij}$ is, the higher frequency $p_{ij}$ is. In other words, two subspaces with large correlation coefficient intend to make similar estimates. Thus, we propose to perform our proposed PCPA algorithm with the collection of subspaces, which has small correlation coefficient $r_S$. 
Simulation results in the next section will verify that these projections result in good error performance.

 \begin{figure*}[htbp]
   \centering
     \subfigure[Frequency matrix]
   {\includegraphics[width=3.2in]{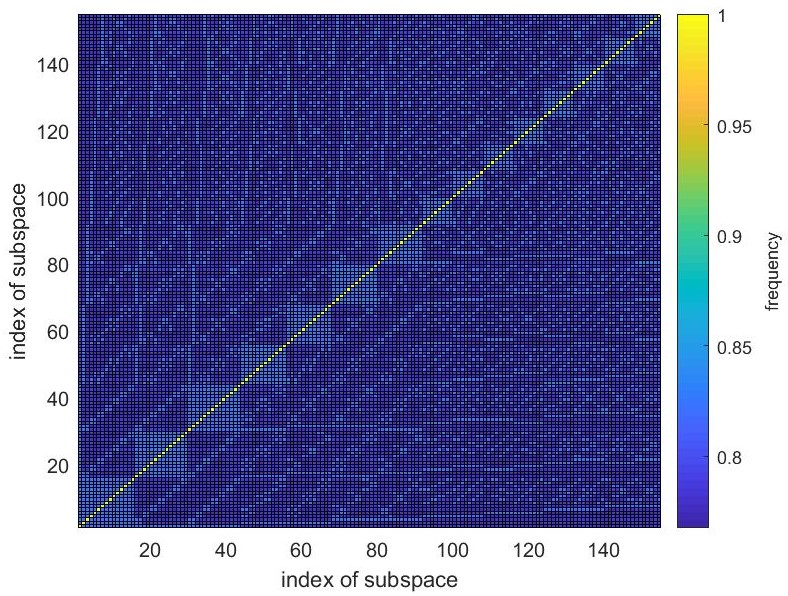}
    \label{fig:first_case}}
     \subfigure[Correlation coefficient matrix]
   {\includegraphics[width=3.2in]{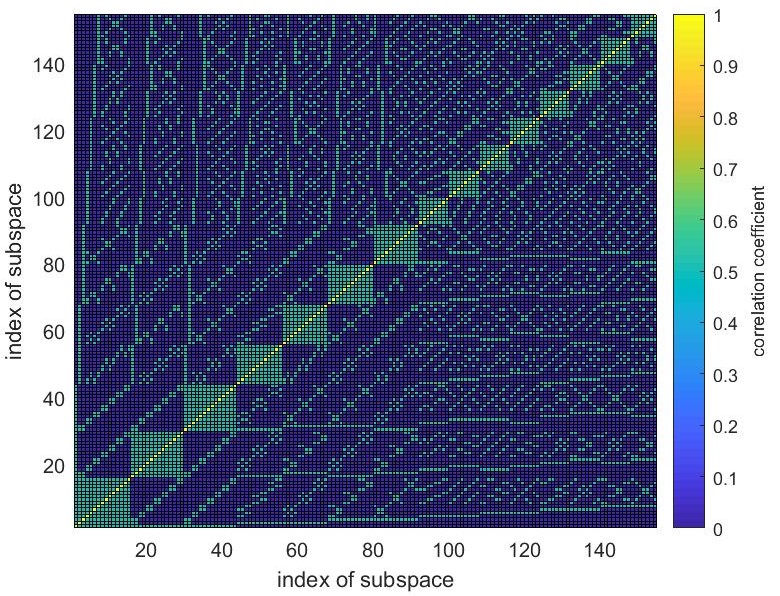}
    \label{fig:second_case}}
   \caption{The frequency matrix and the correlation coefficient matrix of RM$(3,5)$ code.}
   \label{fig:sim2}
 \end{figure*}

\section{Simulation Results}\label{sec:sim}

The PCPA algorithm is tested on RM$(3,5)$ code and RM$(3,7)$ code over the additive white Gaussian noise (AWGN) channel, respectively.
We set the maximal number of iterations $T_{max}=3$ for RPA, CPA and PCPA.
The size of collection which PCPA works with for RM$(3,7)$ is $64$ (denote as PCPA-$64$) and the size of collection which PCPA works with for RM$(3,5)$ is $9$ (denote as PCPA-$9$).
To illustrate that the performance varies with subspace collections,
we choose three different collections in each trial.
Lower union bound is used as a benchmark.

\begin{figure}[htbp]
   \centering
   \includegraphics[width=0.45\textwidth]{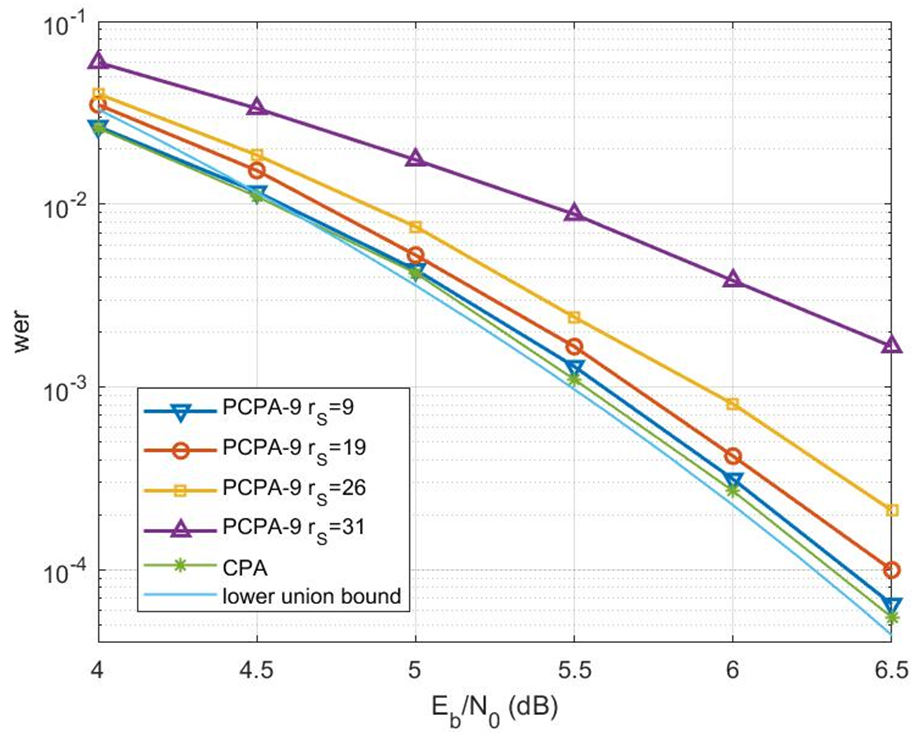}
   \caption{Decoding performance of RM$(3,5)$ code with various algorithms over AWGN channel}
   \label{fig:sim1}
\end{figure}

The simulation results of RM$(3,5)$ are shown in Fig. \ref{fig:sim1}. The PCPA-$9$ decoder with $r_S=9$ significantly outperforms that with $r_S=31$. Moreover, it performs very close to CPA and the lower union bound. In terms of computational complexity, the PCPA-9 decoder for RM$(3,5)$ decodes 9 RM$(1,5)$ codes per iteration, while CPA needs $155$ times.

\begin{figure}[htbp]
   \centering
   \includegraphics[width=0.45\textwidth]{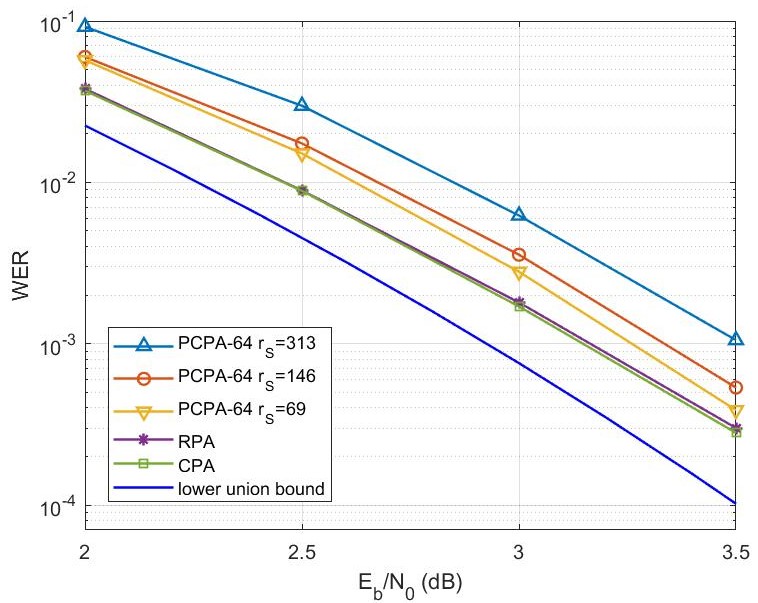}
   \caption{Decoding performance of RM$(3,7)$ code with various algorithms over AWGN channel}
   \label{fig:sim2}
\end{figure}

The simulation results of RM$(3,7)$ are shown in Fig. \ref{fig:sim2}. The PCPA-$64$ decoder with $r_S=313$ performs roughly $0.25$dB worse than the PCPA-$64$ with $r_S=69$ at word error rate (WER) of $10^{-3}$. Besides, the PCPA-$64$ decoder with $r_S=69$ performs only $0.1$dB away from RPA and CPA at WER of $10^{-3}$. In terms of computational complexity, the PCPA-64 decoder for RM$(3,7)$ decodes 64 RM$(1,5)$ codes per iteration, while CPA needs $2667$ times and RPA needs $8001$ times, respectively.

\section{Conclusion}

The paper proposes to collect a few subspaces for projections in terms of our correlation coefficients. Simulation results show that our proposed PCPA is able to approach the performance of the original CPA with only a small collection of subspaces. In short, our proposed PCPA provides a good trade-off between complexity and performance.




\bibliographystyle{IEEEtran}
\bibliography{IEEEabrv_2,template_isit21_2.bib}

\end{document}